\documentclass{manuscript} % my own class :)

\newcommand{\E}{\mathbb{E}}
\newcommand{\Var}{\mathrm{Var}}

\newcommand{\betas}{\beta_{\!S}}
\title{\textbf{A Severity-Aware Reliability Index for Risk-Informed Structural Design}}
\date{\small \today}

\author[1]{\stylizedname{Moussa Leblouba}\thanks{Corresponding author: \href{mailto:mleblouba@sharjah.ac.ae}{mleblouba@sharjah.ac.ae}}}
\author[2]{\stylizedname{Samer Barakat}}
\author[1]{\stylizedname{Raghad Awad}}

\affil[1]{Department of Civil \& Environmental Engineering, College of Engineering, University of Sharjah, P.O.Box 27272, University City, Sharjah, UAE}
\affil[2]{Civil Engineering Department, Fahad Bin Sultan University, Tabuk, Saudi Arabia}

\begin{document}
\maketitle

\begin{abstract}
Classical measures of structural reliability, such as the probability of failure and the related reliability index, are still widely applied in practice. However, these measures are frequency-based only, and they do not give information about the severity of failure once it happens. This missing aspect can cause underestimation of risks, in particular when rare events produce very undesirable consequences. In this paper, a new reliability framework is proposed to address this issue. The framework is based on a new concept, called the Expected Failure Deficit (EFD), which is defined as the average deficiency of the system response when failure occurs. From this quantity, a new reliability index is introduced, called the Severity-Aware Reliability Index, which evaluates the consequence of failure in comparison with the Gaussian benchmark. The mathematical formulation is derived and it is shown that the inverse mapping exists in a restricted domain, which can be interpreted as an indicator of excessive tail risks. A Severity Classification System with five levels is also proposed and calibrated from analytical expressions. Numerical examples, including Gaussian, mildly nonlinear, and heavy-tailed cases, demonstrate that the proposed framework agrees with classical measures in standard situations, while being able to detect hidden severity in more complex cases. The method can therefore be used not only to quantify severity of failure, but also to classify risks in support of engineering design.
\end{abstract}


\section{Introduction}
\label{sec:intro}
The assurance of structural safety and reliability continues to be among the fundamental goals in civil engineering applications. While structural failures remain statistically rare events, their associated consequences may be of significant magnitude. Past failures, such as the collapse of the I-35W Mississippi River Bridge in 2007 or the Sampoong Department Store in 1995, are strong reminders of the catastrophic impacts that a single failure may produce in terms of human casualties, disruption of economic activity, and erosion of public trust.

These tragic events are not remembered due to their occurrence frequency, but rather due to the severity of their outcomes. Such events emphasize a crucial aspect in structural risk evaluations: the acceptability of a system depends not only on how likely the failure is, but also on what kind of failure happens when it occurs. A system that fails frequently but with minor consequences may still be tolerated. In contrast, a system that fails rarely but in a catastrophic manner is typically unacceptable. Therefore, any framework intended for reliability assessment that neglects the consequence aspect of failure may lead to incomplete or even misleading representation of risk.

As a response to these concerns, current structural design practices have adopted reliability-based approaches. In particular, the Load and Resistance Factor Design (LRFD)~\cite{Galambos1981,KOHLER2025102495} methodology has become the most widely implemented framework. In this context, partial safety factors are calibrated in such a way that specific target reliability levels are achieved. These targets are usually expressed in terms of a prescribed probability of failure, $p_f$, or its equivalent reliability index $\beta$~\cite{Galambos1981,Nowak1999,Ellingwood2000}. For instance, for ordinary structural components, the reliability index is commonly set around $\beta=3.5$, which corresponds approximately to a failure probability on the order of ${10}^{-4}$ per year~\cite{Nowak1999,LebloubaKarzadTabshBarakat2022,LebloubaBarakatAltoubatMaalejAwad2022,LebloubaTabsh2020,LebloubaTabshBarakat2020,AWAD2022207}.

Traditional reliability measures, such as the failure probability $p_f$ and the associated index $\beta$, remain central tools in both code calibration and design validation procedures. The index $\beta$ is defined using the standard normal cumulative distribution function as $\beta=-\Phi^{-1}\left(p_f\right)$, and serves as a convenient way to translate probability information into a normalized reliability scale~\cite{ELLINGWOOD2025102474,Cornell1969,ditlevsen1996structural,melchers2018structural,thoft1982structural}. More advanced reliability indices, such as the Hasofer–Lind index, were also developed to achieve geometric invariance in reliability problems~\cite{hasofer1974exact}, yet they still rely on interpreting failure as a binary condition in which a limit-state threshold is crossed.

This binary-type interpretation presents one of the main limitations of classical reliability indices. Whether the structural limit is slightly exceeded or largely violated, both scenarios are assigned the same outcome: a failure has occurred. This means that the classical reliability metrics do not reflect how severe the failure was. For instance, a structure that experiences limited local yielding is treated the same as a complete structural collapse, as long as the same limit-state function is exceeded. This insensitivity to consequence has become a matter of increasing concern, especially in modern infrastructure systems that are subject to rare but extreme events, such as floods, impact loads, or accidental overloads~\cite{SU2025102583}.

In other disciplines, the importance of integrating the severity of outcomes into risk assessments has been acknowledged for several decades. In general risk theory, risk is often expressed as a combination of both the likelihood and the consequence of an undesirable event~\cite{aven2009risk,lowrance1976acceptable}. In structural and insurance applications, expected loss models incorporate both the probability of occurrence and the associated severity~\cite{faber2003risk,kaplan1981on}. In mechanical systems, Failure Mode and Effects Analysis (FMEA) uses a risk metric that depends on the product of severity, occurrence, and detection likelihood~\cite{stamatis2003failure}.

The field of financial risk management has undergone a similar evolution. The classical Value-at-Risk (VaR) measure estimates the probability that a certain threshold is exceeded, without accounting for the magnitude of excess loss~\cite{duffie1997overview}. On the other hand, Conditional Value-at-Risk (CVaR), also referred to as Expected Shortfall, addresses this limitation by considering the expected loss given that the threshold has been exceeded~\cite{rockafellar2000optimization}. CVaR is now widely used because of its desirable mathematical properties and its capacity to account for tail risk in a more appropriate manner~\cite{chaudhuri2020risk,xu2021cvar,filippi2020cvar}.

A common feature among all these examples is that systems or portfolios having equal failure probabilities may exhibit significantly different outcomes. When the analysis is purely based on frequency, both scenarios are considered similar from a reliability point of view. This limitation still remains unresolved in the majority of structural reliability formulations.

In this study, a new reliability index is proposed which explicitly considers the severity aspect of failure. The goal is not to change the classical reliability concept, but to extend it through the introduction of an additional index that accounts for how bad the failure is when it occurs. The idea originates from a natural question: \emph{“If the system under investigation is replaced with an equivalent Gaussian system, then what level of reliability would result in the same average consequence under failure?”}

Answering this question leads to a new severity-aware reliability index, which may be interpreted as a complement to the classical $\beta$ index. The new index is derived through a probabilistic formulation that incorporates the severity of failure events in a systematic manner. The construction starts from a new quantity, introduced in this paper and referred to as the Expected Failure Deficit (EFD). Once normalized and mapped to the Gaussian domain, this quantity provides a meaningful basis for defining severity-aware reliability classes.

This study includes four main contributions:

\begin{enumerate}
\item It introduces the Expected Failure Deficit and the associated severity-aware reliability index, which is benchmarked against an equivalent Gaussian behavior. This formulation extends the classical reliability concept by reflecting not only the frequency of failure but also the level of severity once failure occurs.
\item It derives and validates the mathematical properties of the proposed index, including its monotonicity, continuity, and bounded domain, and includes asymptotic derivations that support its theoretical consistency.
\item It shows that the domain restriction of the proposed index does not represent a drawback but rather acts as a diagnostic feature capable of detecting severe tail behavior, for which equivalent Gaussian interpretation is no longer meaningful.
\item It proposes a new Severity Classification System based on the normalized failure deficit, which defines five severity levels ranging from mild to extreme. The classification thresholds are derived analytically and are linked to classical $\beta$-based benchmarks. This system allows engineers to interpret severity more clearly and supports informed decision-making in structural design.
\end{enumerate}

The remainder of the paper is structured as follows: \cref{sec:concept} presents the mathematical formulation of the Expected Failure Deficit and the corresponding severity-aware index. \cref{sec:theoretical-analysis} analyzes the properties of the transformation and discusses how tail effects are handled. \cref{sec:numerical-investigation} contains numerical applications that demonstrate the behavior of the new index using synthetic and realistic examples. \cref{sec:severity-classification-system} develops the proposed severity classification and integrates it into a practical workflow. \cref{sec:discussion-eng-implications} provides a thorough discussion, implications in engineering design, and directions for future extensions of our work.

% -----------
\section{Conceptual Basis for the Severity-Aware Index}
\label{sec:concept}
In this, we introduce the mathematical formulation that motivates the proposed severity-aware reliability index. First, the classical definitions used in structural reliability are briefly recalled to establish the foundational context. Then, a new quantity, the expected failure deficit, $E_f$, is introduced and normalized to obtain a dimensionless, transferable measure of severity, $E_f^\ast$. This leads to the formulation of a new reliability index, denoted $\betas$, which encapsulates both the frequency and severity of structural failures.

\subsection{Classical Framework and Notation}
\label{subsec:classical-framework}
Let $\mathbf{X}=(X_1,X_2,\ldots,X_n)$ represent a vector of basic random variables describing uncertainties in the structural system, such as loads, material properties, and geometric dimensions. The performance of the structure is modeled through a limit-state function $g(\mathbf{X})$, with $g\left(\mathbf{X}\right)\geq 0$ indicates a safe state and $g(\mathbf{X})< 0$ denotes failure. The failure domain is thus defined as:

\begin{equation*}
\mathcal{F}={\mathbf{X}\in\mathbb{R}^n\mid g(\mathbf{X})<0}.
\end{equation*}

The probability of failure, $p_f$, is given by:

\begin{equation*}
p_f = \mathbb{P}(g(\mathbf{X}) < 0) = \int_{\{x \mid g(x)<0\}} f_{\mathbf{X}}(x) \, dx,
\end{equation*}

where $f_\mathbf{X}(x)$ is the joint density function of $\mathbf{X}$. In many practical problems, the First-Order Reliability Method (FORM) is used to approximate this probability efficiently.

The reliability index $\beta$ is traditionally defined by the inverse of the standard normal cumulative distribution:
\begin{equation*}
\beta=-\Phi^{-1}(p_f),
\end{equation*}

and, in the special case where $g(\mathbf{X})\sim\mathcal{N}(\mu_g,\sigma_g^2)$, the reliability index becomes:

\begin{equation*}
\beta=\frac{\mu_g}{\sigma_g}.
\end{equation*}

\subsection{Expected Failure Deficit and Normalization}
\label{subsec:efd-normalization}
Conventional measures such as $p_f$ and $\beta$ only characterize the frequency of failure, not its magnitude. However, in practice, the depth of failure (i.e., how far into the unsafe region the structure enters) is also critical. To incorporate this, we define the Expected Failure Deficit as:

\begin{equation}\label{eq:efd}
E_f=\mathbb{E}[-g(\mathbf{X})\mid g(\mathbf{X})<0],
\end{equation}

which expresses the average shortfall conditional on failure.

To make this measure scale-invariant and comparable across different systems, it is normalized by the standard deviation of the limit-state function:

\begin{equation}\label{eq:ef}
E_f^\ast=\frac{E_f}{\sigma_g}.
\end{equation}

This normalization parallels the classical reliability index $\beta=\mu_g/\sigma_g$: both scale a mean effect by the global dispersion of $g$. In Gaussian settings, the normalized deficit has a closed form,
\begin{equation}\label{eq:ef-star}
E_f^\ast=\frac{\varphi(\beta)}{\Phi(-\beta)}-\beta,
\end{equation}

with $\varphi$ and $\Phi$ the standard normal density and distribution. The mapping

\[
F(b):=\frac{\varphi(b)}{\Phi(-b)}-b
\]

is strictly decreasing for $b>0$ with

\[
\lim_{b\downarrow 0}F(b)=\frac{2}{\sqrt{2\pi}}\approx 0.7979,
\qquad
\lim_{b\uparrow\infty}F(b)=0,
\]

so $F:(0,\infty) \to (0,2/\sqrt{2\pi})$. 
%Consequently, for any $E_f^\star\in(0,2/\sqrt{2\pi})$ there exists a unique $\betas>0$ solving $F(\betas)=E_f^\star$; if $E_f^\star\ge 2/\sqrt{2\pi}$ the Gaussian benchmark is exceeded and $\betas$ is not defined.

\begin{figure}[t]
  \centering
  \includegraphics[width=0.9\linewidth]{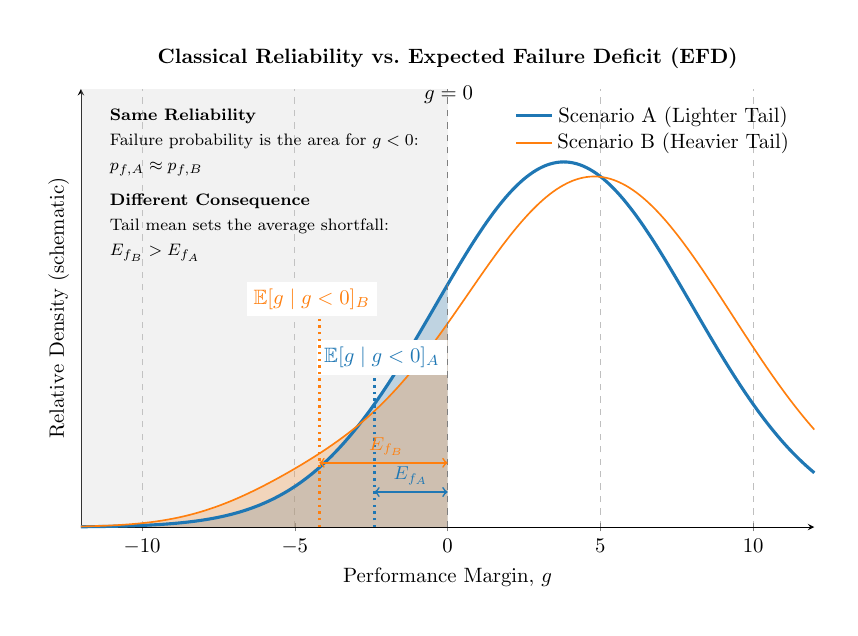}
  \caption{Schematic of the classical failure event and the expected failure deficit (EFD).
  The shaded half-plane is $g<0$ (i.e., $S>R$). The dashed gray line marks $g=0$.
  The tail conditional mean $\mathbb{E}[g\mid g<0]$ locates the center of mass of the failure side, and the horizontal bracket shows $\mathrm{EFD}=\mathbb{E}[-g\mid g<0]$.
  Two systems can share the same $p_f$ (thus the same $\beta$) yet have different EFD.}
  \label{fig:efd-schematic}
\end{figure}

\paragraph{\small Interpretation.}
\cref{fig:efd-schematic} separates \emph{how often} failure occurs (the area $\mathbb{P}\{g<0\}=p_f$, which sets $\beta$ via $p_f=\Phi(-\beta)$) from \emph{how far} failures penetrate on average (EFD). The normalized deficit $E_f^\star$ places $E_f$ on the same scale as the dispersion of $g$; in Gaussian settings it is linked to $b=\mu_g/\sigma_g$ through \cref{eq:ef-star}. Because the Gaussian mapping $F(b)=\varphi(b)/\Phi(-b)-b$ is strictly decreasing with image $(0,\,2/\sqrt{2\pi})$, larger tail severity corresponds to larger $E_f^\star$, approaching the endpoint $2/\sqrt{2\pi}$.

\begin{remark}[Standing assumption]
We assume $g(\mathbf{X}) \in L^2$, i.e., $\sigma_g^2 = \mathrm{Var}[g(\mathbf{X})] < \infty$, so that $E_f^* := \mathbb{E}[-g \mid g < 0]/\sigma_g$ is well-defined. We also tacitly assume $p_f = \mathbb{P}(g(\mathbf{X}) < 0) > 0$.
\end{remark}

\begin{remark}[Robust screening for taily systems]
Because $E_f^\ast$ divides by $\sigma_g$, the index requires $\sigma_g < \infty$. When tail behavior is uncertain, one may first screen with a robust scale (e.g., MAD or IQR) or with the conditional scale $\sigma_{g \mid g < 0}$. This does not alter the main theory; it simply avoids undefined normalizations in infinite-variance regimes.
\end{remark}

\subsection{Monte Carlo illustration: same reliability, different expected failure deficit}
\label{subsec:mc-illustration}

To complement \cref{fig:efd-schematic}, we consider two non-Gaussian scenarios with
\emph{matched} $p_f \approx 10^{-2}$ (thus nearly equal $\beta$) but different tail
severities. In both cases $g=R-S+c$, where the constant shift $c$ is calibrated so that
$\mathbb{P}(g<0)=p_f$. The set-up is:
\begin{itemize}
  \item Scenario A (lighter failure tail): $R\sim$ lognormal; $S\sim$ Gumbel.
  \item Scenario B (heavier failure tail): same $R$; $S\sim$ a mixture of Gumbels with a rare
        extreme component.
\end{itemize}

The results are shown in~\cref{fig:mc-scenario-A,fig:mc-scenario-B}, and the overlays in~\cref{fig:mc-scenario-overlay,fig:mc-scenario-overlay-close}. In each panel, the shaded half-plane is $g<0$; the gray dashed line is $g=0$; the colored
dashed line marks the tail conditional mean $\mathbb{E}[g\mid g<0]$; and the horizontal
double-arrow from that line back to $g=0$ visualizes the expected failure deficit
$E_f=\mathbb{E}[-g\mid g<0]$ (cf.\ \cref{eq:ef}).

\begin{figure}[hbtp]
  \centering
  \includegraphics[width=0.75\linewidth]{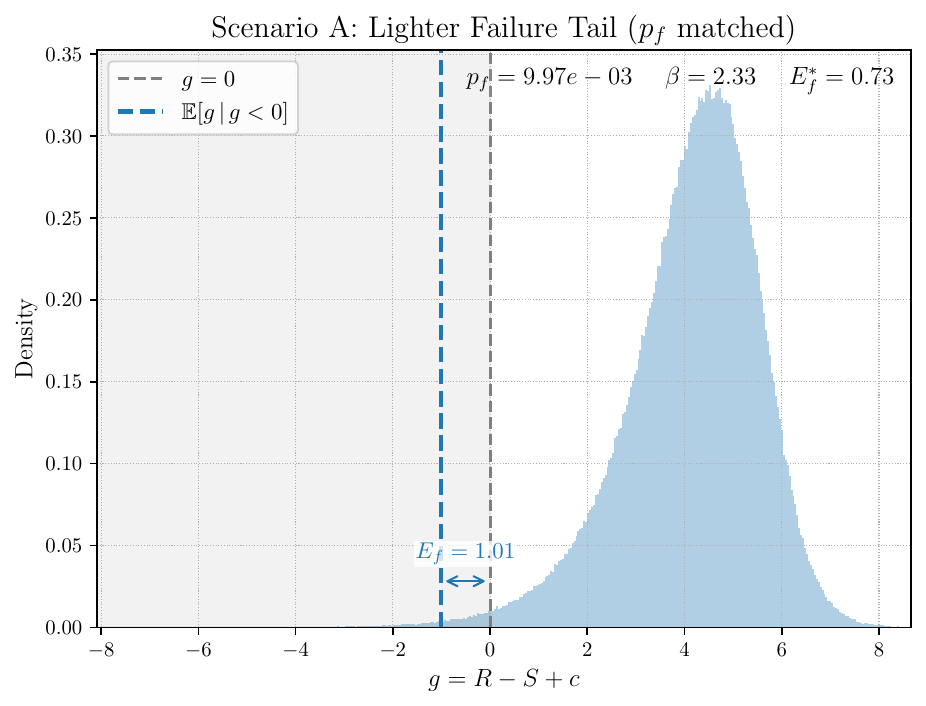}
  \caption{Scenario A (lognormal $R$; Gumbel $S$), calibrated to $p_f\!\approx\!9.97\times10^{-3}$ ($\beta\!\approx\!2.33$). The curve shows the empirical density of $g=R-S+c$; the shaded half-plane marks the failure region $g<0$; the dashed gray line is the boundary $g=0$; the blue dashed line marks the tail conditional mean $\mathbb{E}[g\mid g<0]$; and the blue double-arrow is the expected failure deficit $E_f=\mathbb{E}[-g\mid g<0]$. Here $E_f\!\approx\!1.01$ and $E_f^\ast\!=\!E_f/\sigma_g\!\approx\!0.727$.}
  \label{fig:mc-scenario-A}
\end{figure}

\begin{figure}[hbtp]
  \centering
  \includegraphics[width=0.75\linewidth]{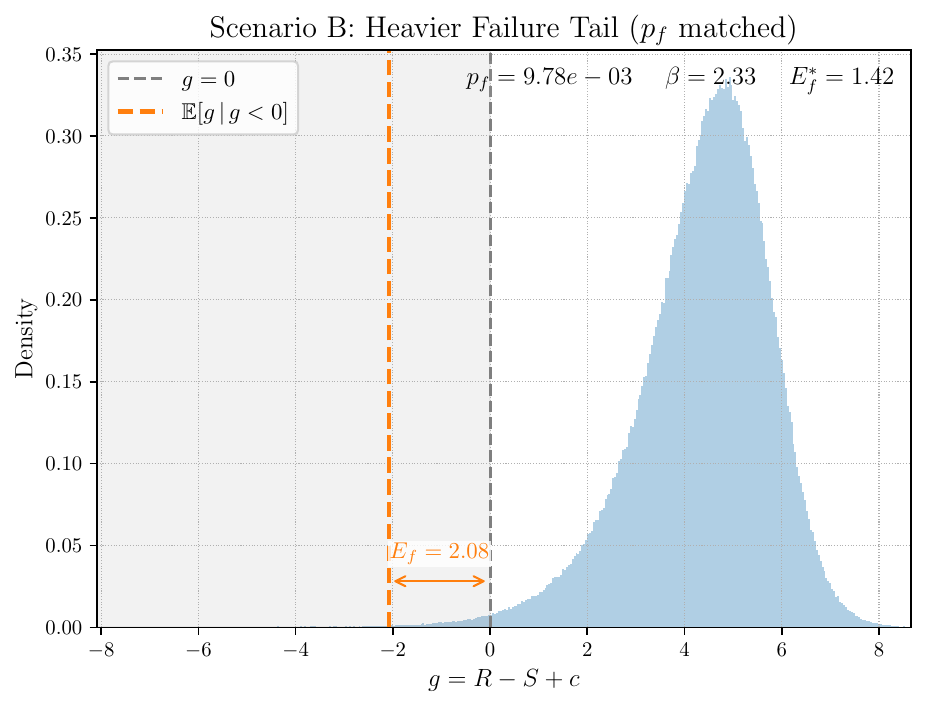}
  \caption{Scenario B (lognormal $R$; mixture-Gumbel $S$ with rare extremes), calibrated to $p_f\!\approx\!9.78\times10^{-3}$ ($\beta\!\approx\!2.33$). Visual elements as in \cref{fig:mc-scenario-A}. The heavier failure tail shifts the tail mean further left, yielding a larger expected failure deficit: $E_f\!\approx\!2.08$ and $E_f^\ast\!\approx\!1.42$.}
  \label{fig:mc-scenario-B}
\end{figure}

\begin{figure}[hbtp]
  \centering
  \includegraphics[width=0.75\linewidth]{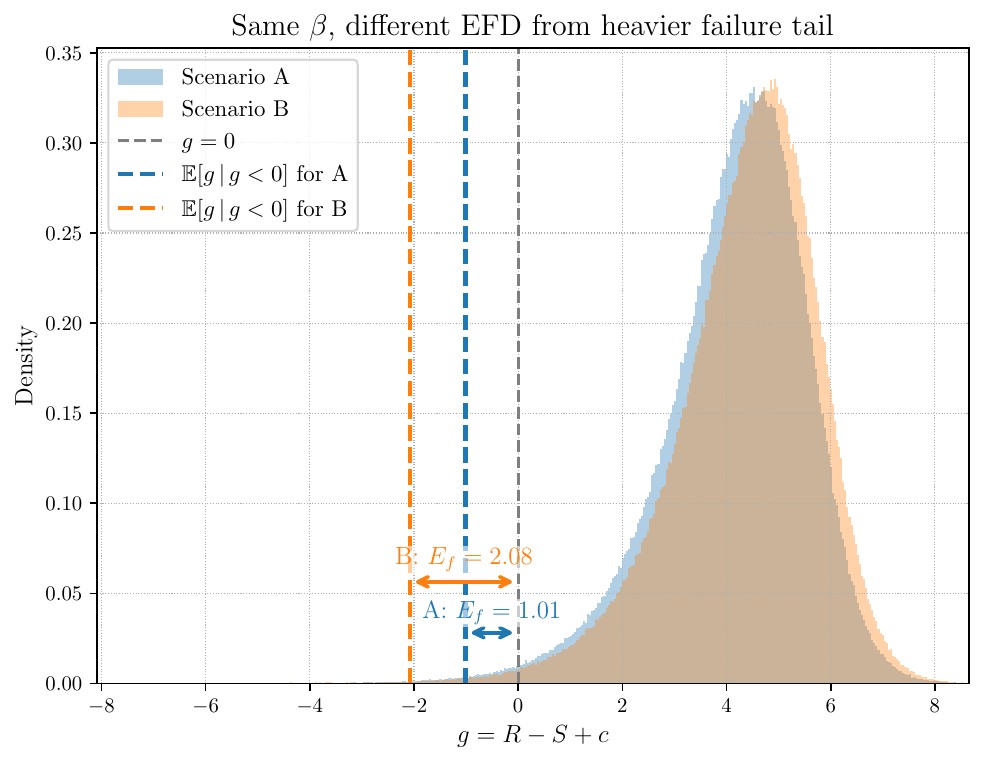}
  \caption{Overlay of Scenarios A (blue) and B (orange) with unified axes and shared binning. Both cases have the same failure probability (hence similar $\beta$), but the heavier tail in Scenario B pulls $\mathbb{E}[g\mid g<0]$ farther left, producing a larger $E_f$ (longer orange bracket). The dashed gray line is $g=0$.}
  \label{fig:mc-scenario-overlay}
\end{figure}

\begin{figure}[hbtp]
  \centering
  \includegraphics[width=0.75\linewidth]{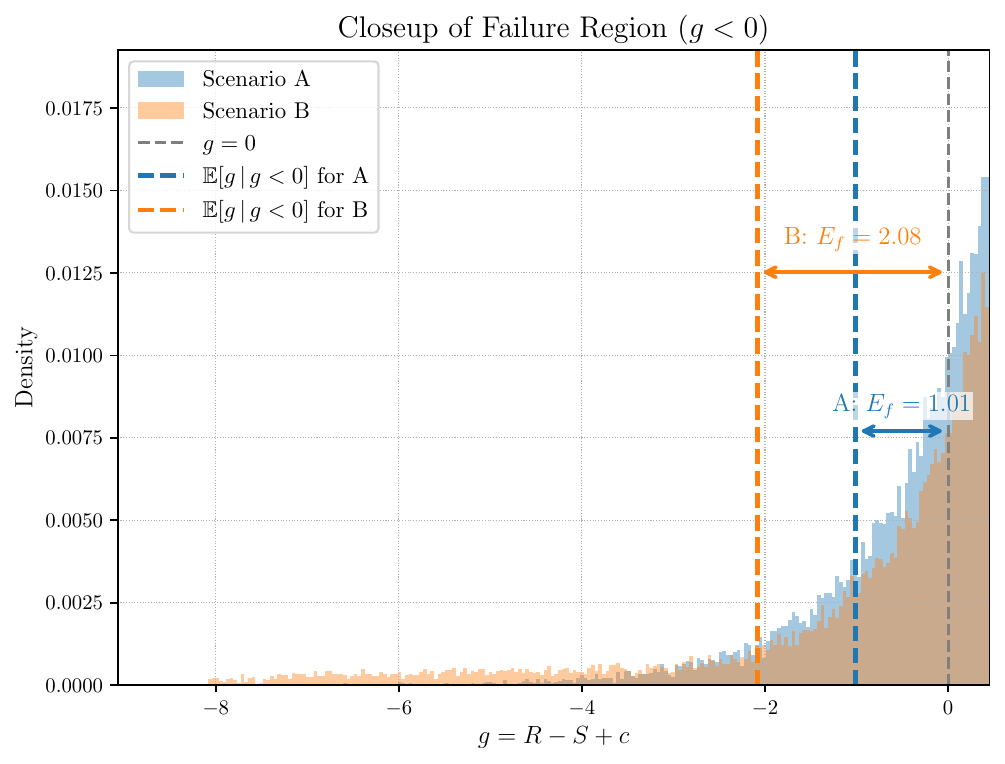}
  \caption{Close-up of the failure region $g<0$ from the overlay: at matched $p_f$ the tail conditional mean lies further into the unsafe side for Scenario B, so $E_f$ is markedly larger than in Scenario A. This isolates the “how far” ($E_{f}$) distinction while holding the “how often” ($p_f$) essentially fixed.}
  \label{fig:mc-scenario-overlay-close}
\end{figure}

\paragraph{\small Numerical results.}
From MC runs with $10^6$ samples per scenario:
\[
\begin{aligned}
&\text{Scenario A: } \mu_g=4.148,\ \sigma_g=1.390,\ p_f=9.971\times 10^{-3},\ \beta=2.327,\\[-1pt]
&\qquad \mathbb{E}[g\mid g<0]=-1.011,\quad E_f=1.011,\quad E_f^\ast=E_f/\sigma_g=0.727;\\[4pt]
&\text{Scenario B: } \mu_g=4.352,\ \sigma_g=1.468,\ p_f=9.775\times 10^{-3},\ \beta=2.335,\\[-1pt]
&\qquad \mathbb{E}[g\mid g<0]=-2.081,\quad E_f=2.081,\quad E_f^\ast=1.418.
\end{aligned}
\]
By construction $\mathbb{P}(g<0)$ is (nearly) the same in A and B, hence $\beta$ is (nearly) identical.
What differs is the \emph{average depth} into the failure region: Scenario~B’s EFD is about
twice Scenario~A’s, visible as the longer orange bracket in the overlay (\cref{fig:mc-scenario-overlay-close}).

\paragraph{\small Relation to the Gaussian benchmark (diagnostic only).}
For Gaussian $g\sim\mathcal N(\mu_g,\sigma_g^2)$,
\[
E_f^\ast=\frac{\varphi(\beta)}{\Phi(-\beta)}-\beta,
\]
so at $\beta\approx 2.33$ the Gaussian benchmark yields $E_f^\ast\approx 0.35$.
Scenario~A’s $E_f^\ast=0.727$ already exceeds this benchmark value (heavier failure tail
than Gaussian at the same $\beta$), while Scenario~B’s $E_f^\ast=1.418$ even exceeds the
Gaussian \emph{endpoint} $2/\sqrt{2\pi}\approx 0.798$, indicating behavior beyond what the
Gaussian mapping can represent.

\paragraph{\small Design interpretation.}
The stress–strength probability $p_f=\mathbb{P}\{S>R\}=\mathbb{P}\{g<0\}$ (and $\beta$) quantifies \emph{how
often} the boundary is crossed; EFD quantifies \emph{how far} failures penetrate once they
occur. For action effects with extreme-value/mixed behavior, two systems with the same
$p_f$ can have very different $E_f$ (and $E_f^\ast$), as the A–B comparison shows. Reporting
$E_f^\ast$ alongside $\beta$ therefore distinguishes designs with identical reliability but
very different expected consequences in failure.

% --------
\subsection{Severity-Aware Reliability Index}
\label{subsec:severity-aware-index}
To incorporate both failure frequency and failure depth into a unified reliability metric, we introduce the Severity-Aware Reliability Index, $\betas$. This new index quantifies the equivalent Gaussian reliability level that would produce the same expected failure deficit as observed in the actual system.

Formally, we define $\betas$ as the unique positive solution to the following nonlinear equation:

\begin{equation}\label{eq:beta-s}
\frac{\varphi\left(\betas\right)}{\Phi\left(-\betas\right)}-\betas=E_f^\ast
\end{equation}

This relationship inverts the analytical expression for $E_f^\ast$ derived under the Gaussian assumption, and it ensures consistency: for systems whose limit-state function is Gaussian, we recover $\betas=\beta$. For other systems, $\betas$ reflects the severity-equivalent reliability index that a Gaussian benchmark would require to yield the same average failure depth.

The defining function:

\begin{equation}\label{eq:Fb}
F(b)=\frac{\varphi(b)}{\Phi(-b)}-b
\end{equation}

is strictly decreasing for $b>0$, guaranteeing that the solution exists and is unique for any admissible value of $E_f^\ast$. This guarantees the well-posedness of the index and enables efficient computation using numerical solvers such as Newton–Raphson or bisection methods.

As a result, $\betas$ serves not only as a severity-aware generalization of the classical reliability index $\beta$, but also as a diagnostic indicator of heavy-tailed or disproportionate failure behavior.

A summary of the core relationships presented in the section is provided in \cref{tab:summary-core-relationships}.

\begin{table}[h!]
\centering
\footnotesize
\caption{Summary of core relationships}
\renewcommand{\arraystretch}{1.75} % Increases row height by 1.5x
\begin{tabular}{>{\raggedright\arraybackslash}p{7.425cm}
                >{\arraybackslash}p{7.425cm}}
\toprule
\textbf{Expression} & \textbf{Description} \\
\midrule
$p_f = \Phi(-\beta)$ & Failure probability \\

$\beta = \dfrac{\mu_g}{\sigma_g}$ & Reliability index (Gaussian $g$) \\

$E_f = \mathbb{E}[-g(\mathbf{X}) \mid g(\mathbf{X}) < 0]$ & Expected failure deficit \\

$E_f^\ast = \dfrac{E_f}{\sigma_g}$ & Normalized deficit \\

$E_f^\ast = \dfrac{\varphi(\beta)}{\Phi(-\beta)} - \beta$ & Closed-form for Gaussian system \\

$\betas : \dfrac{\varphi(b)}{\Phi(-b)} - b = E_f^\ast$ & Severity-aware reliability index \\
\bottomrule
\end{tabular}
\label{tab:summary-core-relationships}
\end{table}

\subsection{Interpretation and Limiting Cases}
\label{subsec:interp-limiting-cases}
The severity-aware index $\betas$ should be interpreted as a back-calibrated reliability index that reflects not just how often failure occurs, but also how serious that failure is on average. A high $\betas$ implies mild or shallow failure behavior, even if $p_f$ is relatively large. Conversely, a low or undefined $\betas$ indicates severe, heavy-tailed risk that classical measures cannot detect.

This distinction gives rise to three conceptual cases:

\begin{description}
\item[{\small Case I (Mild Failure):}] The system fails often, but the failure is shallow. In this case, $\betas>\beta$.
\item[{\small Case II (Gaussian-like Failure):}] The system exhibits failures with frequency and severity aligned. Here, $\betas\approx\beta$.
\item[{\small Case III (Catastrophic Tail Behavior):}] The system rarely fails, but when it does, the consequences are extreme. In this situation, $\betas<\beta$, or may be undefined.
\end{description}

These distinctions will be visualized in \cref{sec:numerical-investigation}.

\section{Theoretical Analysis of the Proposed Reliability Measure}
\label{sec:theoretical-analysis}
In this section we present the mathematical foundation and interpretive analysis of the proposed severity-aware reliability index $\betas$, introduced as a natural extension of the classical reliability index $\beta$. While $\beta$ is uniquely determined by the probability of failure, $\betas$ is defined through the normalized Expected Failure Deficit, $E_f^\ast$, and incorporates information about the conditional magnitude of failure. We aim here to (i) establish the conditions under which $\betas$ is well-defined and unique, (ii) analyze its asymptotic properties and sensitivity, and (iii) recover the classical $\beta$ in the Gaussian limit-state case. The analysis reveals not only the mathematical rigor of the formulation, but also its diagnostic power and interpretability in safety-critical applications.

\subsection{Existence and Uniqueness of $\betas$}
\label{subsec:existence-uniqueness}
Recall that the new reliability measure $\betas$ is defined as the unique positive solution to \cref{eq:beta-s}.

We now formally establish that such a solution always exists and is unique for any normalized expected deficit value $E_f^\ast\in\left(0,2/\sqrt{2\pi}\right)$.
\begin{theorem}[Existence and Uniqueness]
\label{thm:existence-uniqueness}
Let
\begin{equation*}
F(b)=\frac{\varphi(b)}{\Phi(-b)}-b,
\end{equation*}

Then the following properties hold:
\begin{enumerate}
\item $F(b)$ is strictly decreasing on $(0,+\infty)$.
\item $\lim_{b\rightarrow0^+}F(b)=\frac{2}{\sqrt{2\pi}}\approx0.7979$.
\item $\lim_{b\rightarrow+\infty}F(b)=0$.
\item The function $F:(0,\infty)\rightarrow(0,\frac{2}{\sqrt{2\pi}})$ is continuous and bijective.
\item For any normalized expected failure deficit $E_f^\ast\in(0,\frac{2}{\sqrt{2\pi}})$, there exists a unique value of $\betas>0$ such that
\[
F(\betas)=E_f^\ast.
\]
\end{enumerate}

\end{theorem}

\begin{proof}
We prove each part as follows:

\paragraph{\small Strict Monotonicity.}
Let $Z\sim\mathcal N(0,1)$ and $r(b):=\E[Z\mid Z>b]=\varphi(b)/\Phi(-b)$.
Then $F(b)=r(b)-b$ and $r'(b)=r(b)\,[r(b)-b]$ (quotient rule).
Hence
\[
F'(b)=r'(b)-1=r(b)F(b)-1.
\]
For the truncated normal,
\[
\Var(Z\mid Z>b)=1+b\,r(b)-r(b)^2=1-r(b)F(b)>0,
\]
so $r(b)F(b)<1$ and therefore $F'(b)<0$ for all $b>0$.

\paragraph{\small Asymptotic Behavior.}
As \( b \to 0^+ \),

\[
F(b) \to \frac{\varphi(0)}{\Phi(0)} = \frac{1/\sqrt{2\pi}}{0.5} = \frac{2}{\sqrt{2\pi}} \approx 0.7979.
\]

As \( b \to \infty \), by the asymptotic expansion of Mills’ ratio~\cite{GASULL20141832}:

\[
F(b) = \frac{\varphi(b)}{\Phi(-b)} - b \sim \frac{1}{b} + \frac{1}{b^3} + O\left(\frac{1}{b^5}\right) \quad (b \to \infty) \Rightarrow F(b) \sim 0.
\]

\paragraph{\small Continuity and Invertibility.} Because \( F(b) \) is continuous and strictly decreasing on \( (0, \infty) \), and its image is \( \left(0, \frac{2}{\sqrt{2\pi}} \right) \), it admits a unique inverse function 

\[
F^{-1} : \left(0, \frac{2}{\sqrt{2\pi}} \right) \to (0, \infty).
\]

Therefore, for each \( E_f^* \in (0, 0.7979) \), a unique \( \betas \) exists satisfying \( F(\betas) = E_f^* \).
\end{proof}

This result ensures that for any physically meaningful normalized deficit (i.e., $E_f^\ast<0.7979$), a well-defined severity-aware index $\betas$ exists and is unique. Notably, the upper limit of this domain is not merely a mathematical artifact but a fundamental diagnostic bound. If $E_f^\ast$ exceeds this threshold, it signals that the observed severity lies beyond what any Gaussian failure model could explain. Such cases may indicate extreme-tailed behavior, structural brittleness, or regime shifts, each of which may warrant re-evaluation of the model or system design.

\subsection{Asymptotic Behavior of $F(b)$}
\label{subsec:asymptotic-behavior}

Understanding the behavior of $F(b)$ in limiting regimes enhances our interpretation of $\betas$ in extreme cases.

\begin{lemma}[Large $b$ behavior]

As $b\rightarrow\infty$, we have~\cite{GASULL20141832}:

\[
F(b) = \frac{\varphi(b)}{\Phi(-b)} - b \sim \frac{1}{b} + \frac{1}{b^3} + O\left(\frac{1}{b^5}\right) \quad \text{as } b \to \infty
\]

Thus:
\[
\lim_{b \to \infty} F(b) = 0^+
\]
\end{lemma}

This implies that systems with very high reliability (large $\beta$) are associated with extremely small normalized deficits, approaching zero. As expected, very safe systems also fail mildly, if at all.

\begin{lemma}[Small $b$ behavior]
As $b\to 0^+$:

\[
\varphi(-b) \to 0.5, \quad \Phi(b) \to \frac{1}{\sqrt{2\pi}}, \quad \Rightarrow \quad F(b) \to \frac{1}{0.5\sqrt{2\pi}} = \frac{2}{\sqrt{2\pi}}.
\]
\end{lemma}

Hence, this maximum bound on $E_f^\ast$ corresponds to a degenerate system where the performance margin is near zero and violations are large in magnitude on average.

These asymptotics not only validate the domain characterization but also assist in designing robust numerical solvers and preconditioners for computing $\betas$ from given $E_f^\ast$ values.

\subsection{Sensitivity of $\betas$ to Changes in Severity}
\label{subsec:sensitivity-beta-s}

From an engineering standpoint, it is essential to understand how sensitive the severity-aware index is to changes in the normalized failure deficit. This allows for interpreting $\betas$ as a responsive diagnostic quantity.

\begin{theorem}[Sensitivity of $\betas$]

Let $\betas$ satisfy:
\[
F(\betas)=E_f^\ast.
\]

Then:
\[
\frac{d\betas}{dE_f^\ast}=\frac{1}{F^\prime(\betas)}<0.
\]
\end{theorem}

\begin{proof}
By the Implicit Function Theorem, since $F$ is continuously differentiable and $F^\prime(\betas)<0$ on $(0,\infty)$, we obtain:
\[
\frac{d\betas}{dE_f^\ast}=\frac{1}{F^\prime(\betas)}<0.
\]
\end{proof}

The interpretation is immediate and important: as the average severity of failure increases, the corresponding severity-aware reliability index decreases. This behavior is consistent with the desired properties of a risk indicator. Moreover, the magnitude of this derivative may serve as an indicator of local sensitivity: flatter slopes imply that small changes in deficit produce large shifts in perceived reliability, a sign of structural fragility.

\subsection{Consistency with the Gaussian Benchmark}
\label{subsec:consistency-with-gaussian-benchmark}

We now revisit the canonical case where the limit-state function is Gaussian:
\[
g(\mathbf{X})\sim\mathcal{N}(\mu_g,\sigma_g^2).
\]

In this setting:

\begin{itemize}
\item The classical reliability index is given by $\beta=\frac{\mu_g}{\sigma_g}$.
\item The normalized failure deficit is:
\[
E_f^\ast=\frac{\varphi(\beta)}{\Phi(-\beta)}-\beta.
\]
\end{itemize}

But this is precisely the defining equation of $\betas$:
\[
F(\betas)=E_f^\ast=F(\beta)\Rightarrow\betas=\beta.
\]

This equivalence shows that $\betas$ recovers the classical index exactly when Gaussianity holds. Consequently, $\betas$ can be seen as a strict generalization of $\beta$: consistent under the ideal case and diagnostic under deviations.

\subsection{Summary and Implications}
\label{subsec:summary-implications}

The results established in this section demonstrate the mathematical integrity and operational clarity of the new reliability index:

\begin{itemize}
\item The function $F(b)=\frac{\varphi(b)}{\Phi(-b)}-b$ is continuous, strictly decreasing, and maps $(0,\infty)$ to $(0,\frac{2}{\sqrt{2\pi}})$.
\item This mapping ensures that the inverse function defining $\betas$ is unique and well-posed for any normalized failure deficit $E_f^\ast<0.7979$.
\item Asymptotic analysis reveals how $\betas$ behaves in extreme regimes of reliability and severity.
\item Sensitivity analysis confirms that $\betas$ decreases monotonically with increasing severity, thereby making it suitable for practical diagnostics.
\item In the Gaussian case, $\betas$ reduces exactly to $\beta$, confirming backward compatibility and interpretability.
\end{itemize}

\paragraph{{\small Engineering Implication}}
The bounded domain of $F(b)$ is not a limitation, but a diagnostic asset. It signals when systems depart so far from Gaussian behavior (typically via heavy tails or extreme fragility) that even the proposed severity-aware metric reaches its expressive boundary. In such cases, $\betas$ does not fail; it declares the model unfit for conventional calibration, prompting a deeper probabilistic reassessment or new modeling paradigms.

This theoretical foundation justifies the deployment of $\betas$ in practical reliability studies and sets the stage for its application and validation through numerical experiments in the next section.

\section{Numerical Investigation}
\label{sec:numerical-investigation}

This section is devoted to a series of numerical case studies developed to examine the interpretive capability and diagnostic strength of the proposed Severity-Aware Reliability Index, denoted as $\betas$. Each example has been carefully selected to illustrate a distinct aspect of the framework, starting from baseline theoretical consistency, moving to failure profiles characterized by non-critical consequences, and concluding with a configuration that contains rare but extremely severe failure events. The simulations are performed using Monte Carlo sampling, with a total of $N=5,000,000$ realizations per case, thus, ensuring statistically stable estimates for all reliability metrics.
\subsection{Example 1: Gaussian Benchmark -- Validation Under Classical Assumptions}
\label{subsec:example-1}
The first example serves to validate the proposed framework against classical reliability theory under ideal Gaussian conditions. The objective is to demonstrate that the severity-aware index $\betas$ coincides with the traditional reliability index $\beta$ when the underlying assumptions of normality are satisfied.

\subsubsection*{Problem Formulation}

Consider a structural system governed by a linear limit-state function:
\[
g(R,S)=R-S,
\]

where the resistance $R$ and the load $S$ are independently distributed as follows:
\[
R\sim\mathcal{N}(10,1^2),\quad S\sim\mathcal(5,1.5^2).
\]

Given that the sum of independent Gaussian variables is also Gaussian, the performance function $g(R,S)$ is itself normally distributed.

\subsubsection*{Results}

\[
\beta = \frac{\mu_g}{\sigma_g} = \frac{\mu_R - \mu_S}{\sqrt{\sigma_R^2 + \sigma_S^2}} 
= \frac{10 - 5}{\sqrt{1^2 + (1.5)^2}} \approx 2.7735,
\]
\[
E_f^\ast = \frac{\varphi\left(\beta\right)}{\Phi\left(-\beta\right)} - \beta \approx 0.3,
\quad \betas = \beta.
\]

Or using Monte Carlo simulation, which yields almost the same reliability metrics:

\begin{tabular}{@{}l l@{}}
Classical Reliability Index: & $\beta = 2.7748$ \\
Normalized Expected Failure Deficit: & $E_f^\ast = 0.3085$ \\
Severity-Aware Index: & $\betas = 2.6671$
\end{tabular}

\subsubsection*{Interpretation}
As expected, the two indices are in close numerical agreement, with a small discrepancy of approximately 3.9\%. This deviation arises from inherent sampling noise and numerical approximation during root-finding. Importantly, both $\beta$ and $\betas$ indicate a low failure probability and moderate severity. This example confirms the internal consistency of the proposed metric under canonical assumptions and serves as a theoretical validation of the formulation.

\subsection{Example 2: Mild Failure Case -- Distinguishing Favorable Deficit Profiles}
\label{subsec:example-2}

In this example, the performance function remains the same as in the previous case; however, the input distributions are altered to induce non-Gaussian behavior, particularly in the tails. This setup is intended to show how the proposed framework differentiates between systems with frequent but non-catastrophic failures and those with more hazardous profiles.

\subsubsection*{Problem Formulation}
Assuming similar limit-state function $g(R,S)=R-S$, but the random variables are now specified as follows:
\[
R \sim \text{Lognormal}(\mu_{\ln} = 2.3,\, \sigma_{\ln} = 0.2), \quad 
S \sim \text{Gumbel}(\mu = 8,\, \text{scale} = 1.2).
\]
This configuration results in a higher probability of failure due to more dispersed loading, but with a resistance distribution skewed towards higher values, which moderates the severity of the failure events.

\subsubsection*{Results}
Simulation yields:

\begin{tabular}{@{}l l@{}}
Classical Reliability Index: & $\beta = 1.5236$ \\
Normalized Expected Failure Deficit: & $E_f^\ast = 0.3040$ \\
Severity-Aware Index: & $\betas = 2.7219$
\end{tabular}

\subsubsection*{Interpretation}
The classical reliability index $\beta$ indicates a relatively poor performance, corresponding to a failure probability of approximately 6.38\%. However, the severity-aware index $\betas$ is significantly higher, and in fact, closely matches the benchmark value from the Gaussian case in Example 1. This divergence signifies that, although failures occur frequently, they are of considerably less magnitude than would be expected under a Gaussian benchmark. Thus, the $\betas$ index correctly identifies the favorable failure profile and avoids unduly penalizing systems that exhibit moderate deficits. This capability is of practical significance in systems where slight performance exceedances are tolerable or non-critical.

\subsection{Example 3: Extreme Severity Case -- Revealing Hidden Tail Risks}
\label{subsec:example-3}
The third example explores the framework’s diagnostic power in scenarios where rare but extremely severe failure events are embedded within an otherwise benign system. This is a typical configuration where the traditional reliability index may provide misleading assurance of safety due to its insensitivity to tail behavior.

\subsubsection*{Problem Formulation}
The structural model remains $g(R,S)=R-S$, but the load $S$ is modeled as a two-component mixture to capture occasional extreme outcomes:

\[
S \sim 0.999 \cdot \mathcal{N}(5, 2^2) + 0.001 \cdot \text{Pareto}(x_m = 10,\, \alpha = 1.5),
\]

where the nominal component reflects the dominant, moderate loads, and the Pareto component introduces heavy-tailed outliers. The resistance is again normally distributed as $R\sim\mathcal{N}(20,{1.5}^2)$.

This configuration creates a situation where the overall failure probability remains low, but when failure does occur, the resulting deficit is disproportionately large.

\subsubsection*{Results}
Simulation yields:

\begin{tabular}{@{}l l@{}}
Classical Reliability Index: & $\beta=3.3880$ \\
$\sigma_g$ does not exist: & so the $E_f^\ast$ is undefined. \\
Severity-Aware Index: & Not computable (outside valid domain)
\end{tabular}

\subsubsection*{Interpretation}
The classical index $\beta$ indicates a highly reliable system with a failure probability of approximately 0.035\%. However, this assessment is misleading, as evidenced by the extremely high value of $E_f^\ast$, which far exceeds the theoretical upper bound of $\approx 0.7979$ that defines the interpretable domain of the Gaussian-based calibration function. Consequently, $\betas$ is not defined for this case, which is defined as Extreme Severity (infinite variance).

This result is not an anomaly but a critical signal. The framework correctly identifies that failure events, though rare, possess a magnitude so large that they fall outside the conceptual reach of the standard Gaussian-benchmarked model. This property transforms $\betas$ into a diagnostic tool: its incomputability functions as an alert for practitioners that a system may harbor catastrophic risks that evade detection by frequency-based metrics. Such a system demands further investigation and potentially more conservative design or alternative modeling.

\subsection{Summary of Findings}
\label{subsec:summary-findings}

The outcomes of the three numerical studies are compiled in \cref{tab:examples-summary-findings} for comparative analysis:

{\renewcommand{\arraystretch}{1.2} % Adjust row height slightly, if needed

\begin{table}[ht]
\centering
\footnotesize
\caption{Summary of findings}
\label{tab:examples-summary-findings}
\begin{tabularx}{\textwidth}{@{} l >{\raggedright\arraybackslash}X c >{\raggedright\arraybackslash}X @{}}
\toprule
\textbf{Example} & \textbf{System Description} & \textbf{Key Result} & \textbf{Interpretation} \\
\midrule
1. Consistency & Gaussian benchmark & $\betas \approx \beta$ & Confirms alignment with classical reliability theory under ideal conditions. \\
2. Mild Failures & Frequent, low-severity events & $\betas > \beta$ & Indicates favorable failure profile not penalized by the severity-aware metric. \\
3. Extreme Severity & Rare, catastrophic load events & $\beta$ high, $\betas$ undefined & Exposes hidden tail risk missed by traditional metrics; triggers diagnostic warning. \\
\bottomrule
\end{tabularx}
\end{table}}

These findings emphasize that the $\betas$ index is not merely a numerical extension but a conceptual enhancement of classical reliability theory. It retains consistency where appropriate, offers additional resolution in benign scenarios, and, most critically, detects latent high-risk phenomena that may otherwise remain concealed. The framework therefore provides engineers with a more nuanced and informative lens for structural safety assessment.

\subsection{Case Study: Realistic Structural System with High but Quantifiable Severity
}
\label{subsec:case-study-realistic}

To further demonstrate the practical applicability of the proposed severity-aware framework, a case study involving a realistic structural member is examined. This example reflects an engineering design scenario where the system performs adequately under standard loading conditions, yet remains exposed to rare but severe live load events, which significantly elevate the consequences of failure.

\subsubsection{Problem Formulation}

The structural performance is described by the standard limit-state function
\[
g(R,D,L)=R-(1.2D+1.6L),
\]

where $R$ denotes the resistance, $D$ the dead load, and $L$ the live load. The random inputs are defined as follows:

\subparagraph{\small Resistance $R$:}
Modeled by a Lognormal distribution with median 1520 and coefficient of variation 0.10, to capture the positive skew and non-negativity of material strength.
\subparagraph{\small Dead Load $D$:}
Modeled as $D\sim\mathcal{N}(500,{50}^2)$, representing a well-characterized permanent action.
\subparagraph{\small Live Load $L$:}
Represented by a mixture of two Gumbel distributions:
\[
L\sim \left\{\begin{matrix}\mathrm{Gumbel}(\mu=150,\mathrm{scale}=30)&\mathrm{with\ probability\ }0.9995,\\\mathrm{Gumbel}(\mu=500,\mathrm{scale}=30)&\mathrm{with\ probability\ }0.0005.\\\end{matrix}\right.
\]

This formulation reflects typical usage scenarios dominated by moderate live loads, but occasionally perturbed by rare, high-intensity events, such as crowd surges, temporary overloads, or misestimation of service loads.

\subsubsection{Simulation and Results}
A Monte Carlo simulation with $N=2,000,000$ samples was performed to estimate all reliability metrics. The key outcomes are summarized below:

\begin{itemize}
\item \textit{Probability of Failure}: $p_f\approx9.1\times{10}^{-5}$, corresponding to a classical reliability index of $\beta=-\Phi^{-1}(p_f)=3.7442$.
\item \textit{Normalized Expected Failure Deficit}:
\[
E_f^* = \frac{\mathbb{E}\left[ -g(\mathbf{X}) \mid g(\mathbf{X}) < 0 \right]}{\sigma_g} = 0.4741
\]

This value is significantly elevated, reflecting severe failure consequences, but still lies within the admissible domain of the transformation function.
\item \textit{Severity-Aware Reliability Index}: Using the inverse mapping of

\[
\frac{\varphi(b)}{\Phi(-b)}-b=E_f^\ast,
\]
the corresponding severity-aware index is computed as $\betas=1.2777$.
\end{itemize}

\subsubsection{Interpretation and Implications}
This case study reveals a risk scenario in which the classical index paints an overly optimistic picture. The high $\beta=3.74$ reflects a rare probability of failure and would, in a conventional design context, imply excellent reliability. However, this conclusion is misleading, as it fails to account for the consequences of the rare failures.

The severity-aware index $\betas=1.28$, derived from the normalized failure deficit, corrects this illusion by incorporating the magnitude of failure into the assessment. The sharp contrast between $\beta$ and $\betas$ demonstrates that even systems with rare failures may possess an unacceptable severity profile, thereby necessitating reevaluation.

This is not a case of extreme severity, where the deficit lies beyond the definable bounds of the framework. Instead, it is a case of high but quantifiable severity, where the framework not only flags the concern but also returns a meaningful index value. From an engineering standpoint, a value of $\betas=1.28$ would be considered inadequate for critical structures, potentially prompting remedial actions such as structural reinforcement, design redundancy, or refined load control.

Thus, this example demonstrates that the value of $\betas$ not merely as a quantitative reliability index, but as a diagnostic and interpretive tool for uncovering latent vulnerabilities obscured under classical formulations.
\subsection{Visual Diagnostics of Failure Severity Across Risk Profiles}
\label{subsec:visual-diagnostics-of-failure}
To complement the quantitative results presented earlier in \cref{sec:numerical-investigation}, \cref{fig:fig-1-diagnostic-grid} provides a comparative visual analysis of three representative structural reliability scenarios. These were selected to span a broad spectrum of risk behavior, from classical Gaussian systems to ones with mild failure effects and finally to cases with extreme, hidden severity.
Each row in \cref{fig:fig-1-diagnostic-grid} corresponds to one scenario: the left panels show the empirical distribution of the limit-state function $g$, while the right panels display the corresponding failure deficit distributions on a logarithmic scale.

\emph{Top Row (\cref{fig:fig-1-diagnostic-grid}(1a,1b)): Classical Gaussian Case} --- $\beta\approx\betas\approx3.50, E_f^\ast\approx0.54$

\cref{fig:fig-1-diagnostic-grid}(1a,1b) illustrate the baseline Gaussian case. The limit-state distribution in \cref{fig:fig-1-diagnostic-grid}(a) aligns nearly perfectly with the fitted Normal distribution, exhibiting symmetric tails and typical concentration near the mean. In \cref{fig:fig-1-diagnostic-grid}(1b), the distribution of failure deficits (i.e., $-g$ for $g<0$) confirms the absence of extreme severity. The deficit values cluster tightly, with a low mean of 0.54. This scenario demonstrates a situation where both failure frequency and failure magnitude are low and well-characterized, resulting in consistency between the classical reliability index $\beta$ and the severity-aware index $\betas$.

\begin{figure}[htbp] %  figure placement: here, top, bottom, or page
   \centering
   \includegraphics[width=1.0\textwidth]{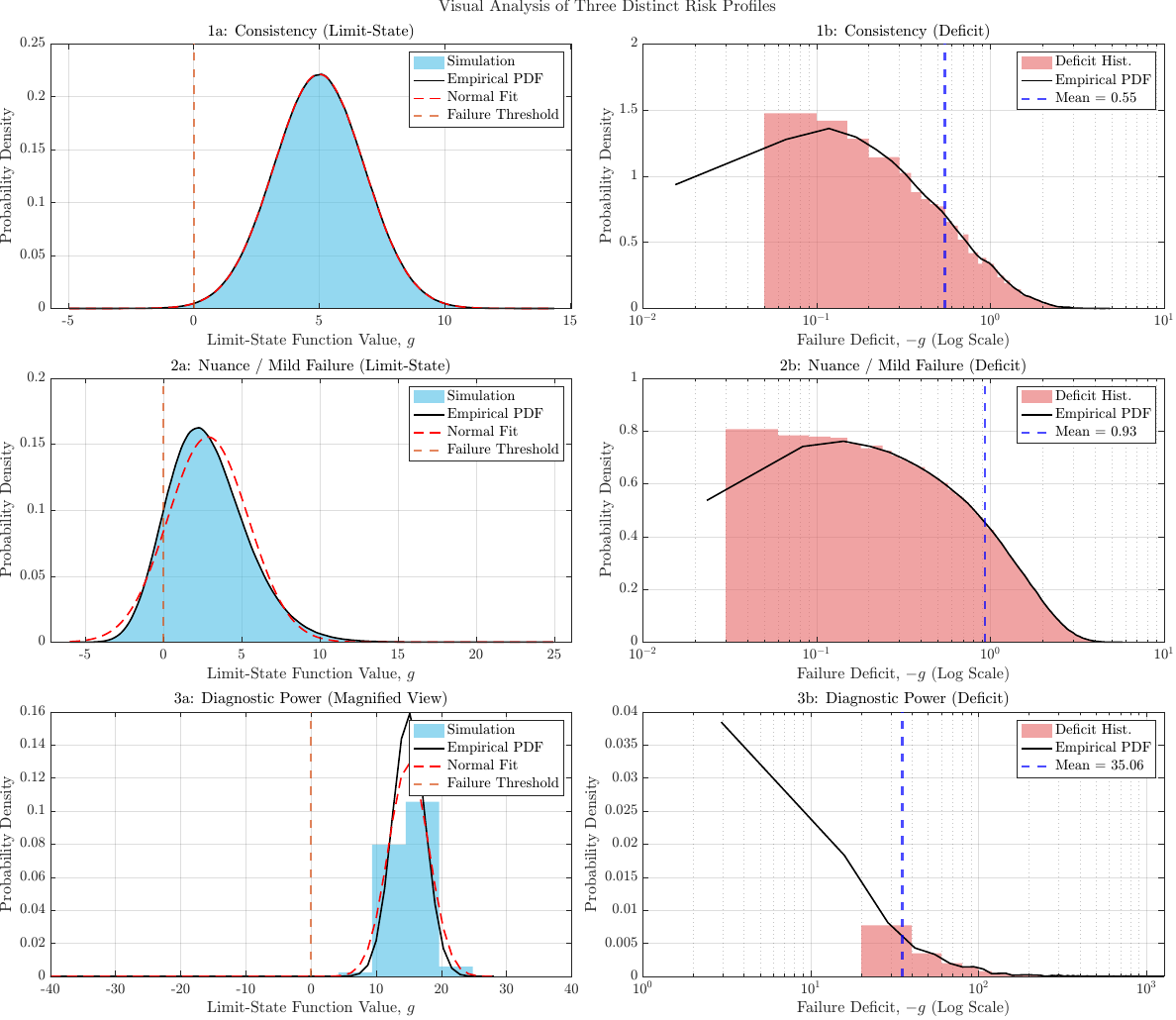} 
   \caption{Visual comparison of three structural risk scenarios, using both limit-state function g distributions (left column) and corresponding failure deficit histograms for $g<0$ (right column, log-scale). Top row (1a–1b): Gaussian benchmark case -- $\beta\approx\betas\approx3.50$, $E_f^\ast\approx0.251$; symmetric tails and compact deficits yield consistent indices. Middle row (2a–2b): Mild failure case – $\beta\approx3.11$, $\betas\approx3.84$, $E_f^\ast\approx0.235$; failure frequency increases, but severity remains bounded. Bottom row (3a–3b): Heavy-tailed case – $\beta\approx3.14$, $\betas$ undefined, $E_f^\ast=1.67>0.7979$; catastrophic deficits cause diagnostic threshold violation. These visualizations illustrate the ability of the severity-aware framework to differentiate between benign and catastrophic failure behaviors, even when classical reliability metrics appear similar.}
   \label{fig:fig-1-diagnostic-grid}
\end{figure}

\emph{Middle Row (\cref{fig:fig-1-diagnostic-grid}(2a,2b)): Mild Failure Case} --- $\beta\approx3.11, \betas\approx3.84, E_f^\ast\approx0.23$

In this case, the system exhibits a higher failure probability, as seen in the denser left tail of \cref{fig:fig-1-diagnostic-grid}(2a). Although the Normal approximation remains acceptable near the mean, deviations emerge near the threshold. \cref{fig:fig-1-diagnostic-grid}(2b) reveals that these failures are consistently moderate: the deficit values remain bounded within a narrow range, with a mean of 0.93 and no significant outliers. This combination of higher frequency but relatively low consequence explains why $\betas>\beta$ here—the failures, while more likely, are less severe than a Gaussian failure assumption would suggest. The severity-aware framework interprets this correctly by assigning a higher index.

\emph{Bottom Row (\cref{fig:fig-1-diagnostic-grid}(3a,3b)): Heavy-Tailed Case} --- $\beta\approx3.14$, $\betas$ undefined, $E_f^\ast=1.67>0.7979$

This final case exposes the diagnostic strength of the severity-aware index. The main plot in \cref{fig:fig-1-diagnostic-grid}(3a) gives the misleading impression of a symmetric and safe distribution. However, the inset reveals the hidden tail extending deeply into the negative domain, a heavy-tailed behavior not captured by the Gaussian fit. \cref{fig:fig-1-diagnostic-grid}(3b) further demonstrates this: the deficit values span several orders of magnitude, with a mean of 42.16. The normalized deficit $E_f^\ast$ exceeds the theoretical diagnostic limit of 0.7979, indicating that no equivalent Gaussian system could explain this severity level. The severity-aware index $\betas$ is therefore undefined, not due to a computational failure, but as a deliberate signal that the system’s risk lies outside the Gaussian paradigm. The classical $\beta$, while seemingly adequate, underestimates the catastrophic nature of possible outcomes.

\section{A Severity Classification System for Risk-Informed Design}
\label{sec:severity-classification-system}
The previous sections have established that the Severity-Aware Reliability Index, $\betas$, supplements the classical reliability index $\beta$ by introducing a measure that accounts for the depth or seriousness of failure consequences. Whereas $\beta$ captures how frequently failure occurs, it cannot distinguish between a ductile, localized failure and a brittle, catastrophic one, even when both have equal probability. This section translates the abstract notion of normalized expected failure deficit, $E_f^\ast$, into a formal, interpretable severity classification system suitable for practical use in design and risk appraisal.

\subsection{Derivation of Non-Arbitrary Severity Thresholds}
\label{subsec:derivation-severity-thresholds}
A recurring criticism of novel metrics is the use of seemingly arbitrary thresholds for categorization. In this work, such arbitrariness is explicitly avoided by anchoring severity levels to reliability index values already familiar to practicing engineers. Instead of choosing cutoffs heuristically, thresholds are algorithmically derived by mapping classical reliability values to their corresponding failure deficit using the inverse transformation:

\[
E_f^\ast=F(\betas)=\frac{\varphi(\betas)}{\Phi(-\betas)}-\betas,
\]

This function is monotonic, invertible, and approaches a finite upper bound as $\betas \to 0$, with:

\[
\lim_{\betas \to 0^+} E_f^\ast = \frac{2}{\sqrt{2\pi}} \approx 0.7979.
\]

The thresholds for the classification levels are computed by evaluating $E_f^\ast$ at reliability indices commonly used in practice (see \cref{fig:fig-2-F-function}):

\begin{align*}
\betas &= 3.0 \quad \text{(standard target for high-reliability structures)} 
&& \Rightarrow \quad E_f^\ast \approx 0.283 \\
\betas &= 2.0 \quad \text{(boundary for moderate-risk components)} 
&& \Rightarrow \quad E_f^\ast \approx 0.373 \\
\betas &= 1.0 \quad \text{(rarely acceptable in structural design)} 
&& \Rightarrow \quad E_f^\ast \approx 0.525
\end{align*}

This procedure guarantees that all category boundaries are reproducible, transparent, and theoretically sound.
 
\begin{figure}[htbp] %  figure placement: here, top, bottom, or page
   \centering
   \includegraphics[width=0.8\textwidth]{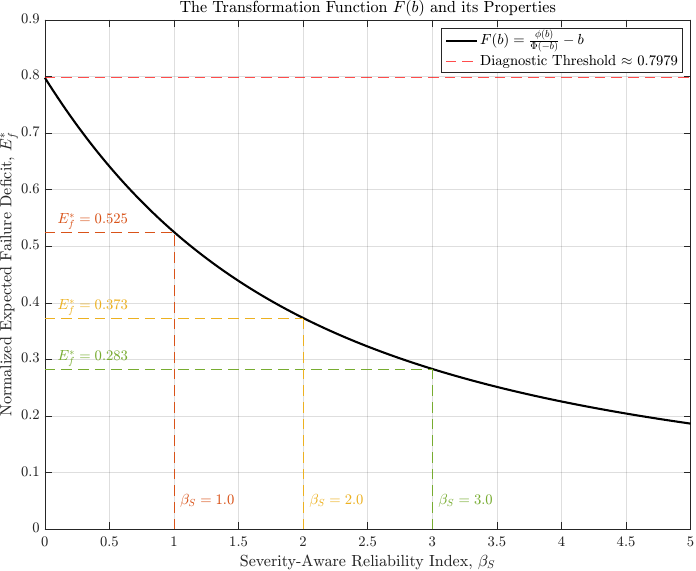} 
   \caption{Severity thresholds and their corresponding expected failure deficits}
   \label{fig:fig-2-F-function}
\end{figure}

\subsection{Five-Level Severity Classification System}
\label{subsec:five-levels-severity-classification}

Based on the derived thresholds, we now propose a five-level severity scale (\cref{tab:severity-classification-system}). The classification allows engineers to translate a single computed value of $E_f^\ast$ into a qualitative risk interpretation, while preserving quantitative rigor.

\begin{table}[h!]
\centering
\footnotesize
\caption{Proposed Severity Classification System}
\begin{tabular}{>{\raggedright\arraybackslash}p{2.2cm}
                >{\centering\arraybackslash}p{2.75cm}
                >{\centering\arraybackslash}p{3.5cm}
                >{\arraybackslash}p{5.9cm}}
\toprule
\textbf{Severity Level} & \boldmath{$\betas$ \textbf{Range}} & \boldmath{$E_f^\ast$ \textbf{Range}} & \textbf{Interpretation and Recommended Action} \\
\midrule
I: Mild & $\betas \geq 3.0$ & $0 < E_f^\ast < 0.283$ & Failure is gentle. Classical reliability index governs. No further mitigation required. \\
II: Moderate & $2.0 \leq \betas < 3.0$ & $0.283 \leq E_f^\ast < 0.373$ & Severity is perceptible. Enhanced quality assurance and optional monitoring may be prudent. \\
III: High & $1.0 \leq \betas < 2.0$ & $0.373 \leq E_f^\ast < 0.525$ & Severity is non-negligible. Consider design reinforcement or redundancy. \\
IV: Critical & $0 < \betas < 1.0$ & $0.525 \leq E_f^\ast < 0.7979$ & System approaches the risk boundary. Strengthening or redesign should be carried out. \\
V: Extreme & Incomputable & $E_f^\ast \geq 0.7979$ & Catastrophic severity. Risk is beyond acceptable domain. Conceptual overhaul is necessary. \\
\bottomrule
\end{tabular}
\label{tab:severity-classification-system}
\end{table}

The classification system can be applied across structural types and performance targets. Its interpretive power lies in exposing not only the frequency of failure, but also how physically damaging that failure is likely to be.

\subsection{Relationship with Code-Based Risk Categorization}
\label{subsec:relationship-with-code-risk-categorization}

Current structural design codes, such as ASCE 7, already implement a risk-based framework by assigning structures to Risk Categories I through IV. These categories reflect the societal consequences of failure (e.g., a hospital vs. a warehouse), and influence the design via importance factors. These factors increase the applied loads and thereby elevate the required target reliability index $\beta$.

However, this system is qualitative and frequency-oriented. It ensures that failure becomes rarer, but not necessarily less damaging. Once failure occurs, ASCE 7 provides no means to assess whether the structural response is ductile and localized or brittle and widespread.

The severity classification system proposed in this study addresses this limitation directly. Rather than substituting existing design provisions, it offers a complementary diagnostic tool. \cref{tab:comparison-risk-categories} compares the two frameworks.

\begin{table}[h!]
\centering
\footnotesize
\renewcommand{\arraystretch}{1.5}
\caption{Comparison of Risk Categorization Approaches}
\begin{tabular}{>{\raggedright\arraybackslash}p{3.5cm}
                >{\raggedright\arraybackslash}p{5.2cm}
                >{\arraybackslash}p{6.15cm}}
\toprule
\textbf{Feature} & \textbf{ASCE 7 Risk Categories} & \textbf{Severity-Based Classification (This Work)} \\
\midrule
Basis of Classification & Qualitative (use and occupancy) & Quantitative (failure deficit, $E_f^\ast$) \\
Targeted Property & Importance-driven reliability & Failure consequence depth \\
Mechanism & Indirect (load amplification) & Direct (severity evaluation) \\
Threshold Source & Prescriptive, code-defined & Computed from inverse Gaussian mapping \\
Tail Risk Visibility & Not visible & Explicitly detected when $E_f^\ast \geq 0.7979$ \\
\bottomrule
\end{tabular}
\label{tab:comparison-risk-categories}
\end{table}

In practice, a two-step design workflow is recommended:
\begin{enumerate}
\item \textbf{Frequency Evaluation}: Use ASCE 7 (or relevant standard) to determine required $\beta$ based on Risk Category. Ensure design meets this requirement.
\item \textbf{Severity Evaluation}: Compute $E_f^\ast$ for the proposed design. Use \cref{tab:severity-classification-system} to classify severity and determine if redesign or mitigation is warranted.
\end{enumerate}

For non-critical structures, a severity classification of III or lower may be acceptable. For critical infrastructure, II or lower should be required. This dual assessment ensures that a design is not only reliable in the classical sense, but also resistant to high-consequence failures, therefore, satisfying both the frequency and the impact dimensions of structural risk.

\subsection{A Practical Workflow for Integrated Frequency-Severity Assessment}
\label{subsec:practical-workflow-severity}

To synthesize the severity classification framework with conventional design practices, this subsection presents a systematic workflow for severity-aware reliability assessment. The purpose is to guide engineers through a two-tiered validation process: first, a frequency-based verification aligned with existing code provisions; second, a severity-based diagnosis enabled by the newly proposed metric $E_f^\ast$. The decision logic is summarized in the flowchart shown in \cref{fig:fig-3-workflow}.

\begin{figure}[h!] %  figure placement: here, top, bottom, or page
   \centering
   \includegraphics[width=0.7\textwidth]{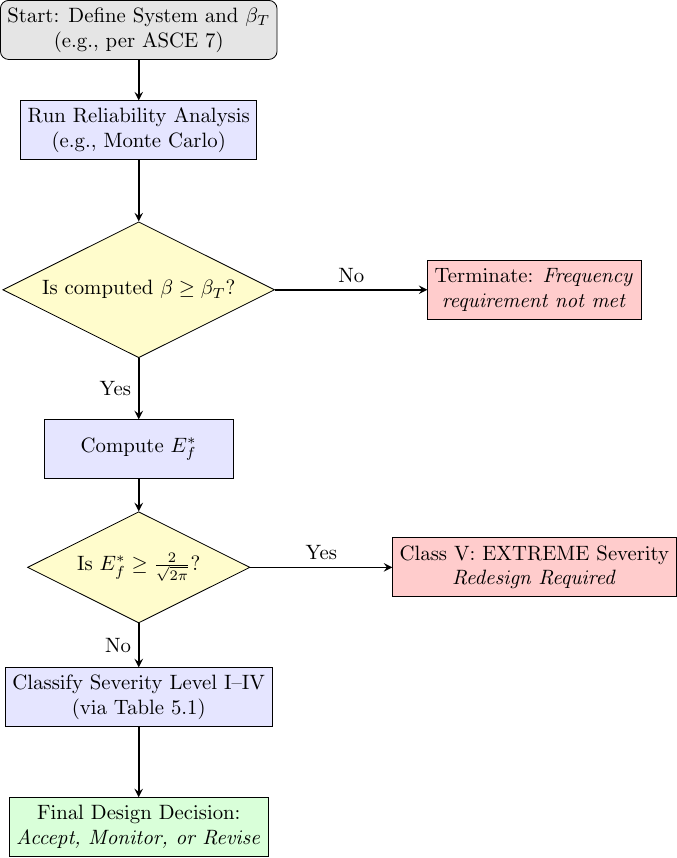} 
   \caption{Workflow diagram for integrated reliability assessment combining traditional frequency metrics with severity-aware classification.}
   \label{fig:fig-3-workflow}
\end{figure}

The workflow begins with the standard practice of establishing a target reliability index, $\beta_T$, according to code-based requirements (e.g., from ASCE 7 Risk Category assignments). Following this, a comprehensive reliability analysis is performed using appropriate probabilistic tools, such as Monte Carlo simulation.

The first validation checkpoint is a frequency check. Here, the computed classical reliability index, $\beta$, is compared against the prescribed target. If this requirement is not satisfied, the design is rejected, as it fails to meet the minimal frequency standard mandated by prevailing codes.

Only if the frequency condition is satisfied does the process advance to the severity check. At this stage, the normalized expected failure deficit, $E_f^\ast$, is calculated, capturing the depth of potential failure excursions. The pivotal diagnostic threshold then evaluates whether $E_f^\ast$ exceeds the theoretical upper limit of the Gaussian-calibrated domain. If it does, the design is classified as Level V: Extreme, indicating catastrophic, heavy-tailed risk that eludes detection by classical reliability indices. Such systems must be re-conceptualized or re-engineered entirely.

If the computed $E_f^\ast$ lies below this critical bound, the severity level is then classified (Levels I-IV) using the ranges established in \cref{tab:severity-classification-system}. This step implicitly maps the value of $E_f^\ast$ to the corresponding severity-aware reliability index, $\betas$. Finally, a holistic final design decision can be rendered based on both the successful frequency check and the now-quantified severity level. This may involve acceptance of the design, implementation of targeted monitoring strategies, or execution of further design refinements.

This structured process ensures that a design is not evaluated solely on the likelihood of failure but also on the magnitude of its consequences, hence delivering a more complete and risk-informed foundation for structural safety.

% -----------------------------------------------------
% 		Discussion and Engineering Implications
% -----------------------------------------------------

\section{Discussion and Engineering Implications}
\label{sec:discussion-eng-implications}
This section provides an expanded interpretation of the new reliability index $\betas$, developed from the normalized Expected Failure Deficit $E_f^\ast$. We contrast its behavior with classical reliability metrics, emphasize its implications for engineering safety, and discuss the limitations and future prospects of this framework in practical risk assessment.

\subsection{Comparison with Classical Reliability Index}
\label{subsec:comparison-with-classical-rel-index}
Traditional structural reliability assessment relies on the probability of failure $p_f$, or the equivalent reliability index $\beta=-\Phi^{-1}(p_f)$, to quantify risk. These metrics capture how often a structure is expected to fail, assuming implicitly that all failures are of similar consequence. This modeling assumption simplifies calibration but disregards how severe a failure may be when it does occur.

The severity-aware index $\betas$, introduced in this study, addresses this shortcoming. It is constructed through the normalized conditional failure deficit:

\[
E_f^\ast=\frac{\mathbb{E}[-g(\mathbf{X})\mid g(\mathbf{X})<0]}{\sigma_{g}},
\]

and defined implicitly by solving the transformation:

\[
\frac{\varphi(\betas)}{\Phi(-\betas)}-\betas=E_f^\ast.
\]

This formulation preserves interpretability under Gaussian conditions, where $\betas=\beta$, but detects deviations in tail behavior when failure severity is extreme. As demonstrated in \cref{sec:numerical-investigation}, systems with identical $p_f$ but different deficit profiles yield markedly different $\betas$ values. In the heavy-tailed case, we observed that although $\beta\approx3.14$, the corresponding $E_f^\ast=1.67$ exceeds the diagnostic threshold 0.7979, rendering $\betas$ undefined. This contrast illustrates the inability of classical metrics to reflect rare but catastrophic scenarios.

The classical $\beta$ gives a frequency-based impression of safety, but $\betas$ complements it with a severity-adjusted perspective. Together, they form a richer picture of structural risk.

\subsection{Practical Implications in Design and Decision-Making}
\label{subsec:practical-implications-in-design}
The integration of failure severity via $\betas$ has several practical implications:

\begin{description}
\item[{\small Dual Risk Quantification:}] $\betas$ enables simultaneous consideration of both failure probability and failure consequence, bridging a long-standing gap in reliability modeling. This is especially critical in safety-sensitive applications.

\item[{\small Early Warnings in Tail-Dominated Systems:}] Systems may appear safe under $\beta$, but an undefined or reduced $\betas$ may signal dangerous tail behavior. This diagnostic capability serves as an early warning that traditional analyses might miss.

\item[{\small Rationalizing Redundancy and Design Margins:}] A structure with $\beta\geq3$ but low or undefined $\betas$ may warrant reconsideration. Engineers could use this signal to adjust safety factors or introduce mitigation strategies for extreme outcomes.

\item [{\small Reform of Design Standards:}] Risk-based codes that incorporate only $p_f$ may be blind to rare, high-impact events. The proposed framework offers a feasible extension pathway by preserving compatibility with classical benchmarks.

\item [{\small Severity Classification:}] \cref{sec:severity-classification-system} proposed a severity classification system (Mild, Moderate, Severe, and Extreme) based on ranges of $E_f^\ast$. This classification offers a natural supplement to $\betas$ and provides actionable insight: systems with Extreme severity ($E_f^\ast>0.7979$) cannot be reliably characterized with Gaussian tools and may require fundamental design reevaluation.

\end{description}

\subsection{On the Invertibility and Diagnostic Value of $\betas$}
\label{subsec:invertibility-diagnostic-beta-s}
By construction, the transformation defining $\betas$ is only invertible for
\[
0<E_f^\ast<\frac{2}{\sqrt{2\pi}}\approx0.7979.
\]

This bound is not a limitation but a feature: it identifies systems whose severity profiles cannot be matched by any Gaussian distribution. If a system yields $E_f^\ast>0.7979$, the corresponding $\betas$ becomes undefined, not due to computational failure, but because no equivalent Gaussian model can be constructed. Such cases require special scrutiny.
The lack of a solution indicates a departure from conventional risk models. In such cases, engineers may consider adopting alternative characterizations such as:

\begin{itemize}
\item Truncated or heavy-tailed distributions.
\item Nonparametric tail estimators.
\item Conservative assignments (e.g., $\betas \to 0$) to reflect unquantifiable severity.
\end{itemize}

This diagnostic capacity distinguishes $\betas$ from all classical indices.

\subsection{Directions for Generalization}
\label{subsec:directors-for-generalization}
While $\betas$ retains theoretical and practical consistency under Gaussian calibration, its application in broader contexts, such as skewed, heavy-tailed, or compound failure models, may require generalization.
Future work may consider:

\begin{itemize}
\item Extending the transformation using alternative base distributions (e.g., Student-t, Laplace).
\item Developing hybrid formulations that interpolate between Gaussian-based $\betas$ and tail-aware surrogates.
\item Exploring severity metrics based on tail-weighted moments or conditional risk integrals.
\end{itemize}

Such developments would widen the expressive range of the framework while preserving its diagnostic rigor.

\subsection{Concluding Perspective}
\label{subsec:concluding-perspective}
The severity-aware reliability index $\betas$ offers a powerful yet interpretable generalization of classical reliability metrics. Its ability to remain consistent in the Gaussian regime and deviate meaningfully in non-Gaussian scenarios makes it particularly suitable for modern structural safety evaluations. When combined with the severity classification system proposed in \cref{sec:severity-classification-system} and supported by visual diagnostics (\cref{sec:numerical-investigation}), it provides a comprehensive and actionable view of structural risk.

Importantly, the infeasibility of computing $\betas$ for large $E_f^\ast$ values is not a limitation; it is an indicator of profound risk mischaracterization under standard assumptions. This interpretive distinction equips engineers with a valuable signal that may otherwise go unnoticed in traditional assessments.

Through this diagnostic framing, $\betas$ shifts the reliability paradigm from binary limit-state violation to a richer, consequence-informed view of safety, one that is more aligned with the actual stakes of structural performance.

% -----------------------------------------------------
% 						Conclusions
% -----------------------------------------------------

\section{Conclusions}
\label{sec:conclusions}
This study introduced a new reliability index, $\betas$, designed to account not only for the probability of structural failure but also for the severity of such failure. Built upon the normalized Expected Failure Deficit $E_f^\ast$, the proposed measure provides a more informative representation of structural safety than the classical reliability index $\beta$, which captures frequency but ignores consequences.

It has been shown that $\betas$ is consistent with the classical index in the Gaussian setting, where both measures coincide exactly, thus preserving theoretical compatibility under standard assumptions. However, in systems with skewed or heavy-tailed behavior, the new index exhibits sensitivity to high-consequence outliers that would remain hidden under conventional reliability models.

Our mathematical analysis established that $\betas$ is well-defined for $E_f^\ast\in(0,0.7979)$, a domain dictated by the invertibility of the defining transformation $F(b)=\frac{\varphi(b)}{\Phi(-b)}-b$. This upper bound is not an artifact but a diagnostic threshold: exceeding it implies that failure severity exceeds the descriptive capacity of Gaussian-based metrics. In such cases, the absence of a valid $\betas$ solution reflects unquantifiable structural risk and signals that further modeling or re-design is warranted.

The numerical demonstrations emphasized that small failure probabilities can be misleading when associated with high conditional shortfalls. For instance, examples where $\beta\approx3$ but $\betas$ either dropped sharply or became undefined demonstrate a critical insight: the absence of frequent failure does not guarantee safety. The severity-aware index successfully identified these risks and offered early warnings that traditional methods could not provide.

To summarize the key contributions of this work:

\begin{itemize}
\item A new reliability index $\betas$ incorporating both frequency and conditional failure severity, retaining theoretical consistency in Gaussian cases;
\item A full analytical characterization of the transformation from $E_f^\ast$ to $\betas$, including its invertibility limits and diagnostic behavior;
\item Introduction of a severity classification system to support practical interpretation of $E_f^\ast$ values, particularly in identifying systems operating in extreme regimes;
\item Case studies demonstrating the divergence between $\beta$ and $\betas$ and revealing how the proposed metric exposes vulnerabilities obscured by conventional analysis.
\end{itemize}

Future research may aim to generalize the transformation framework to accommodate broader tail behaviors, enhance tail risk estimation, and examine how $\betas$ may be embedded in reliability-informed design optimization, risk-based codes, or multi-hazard performance criteria.

Altogether, this work demonstrates that risk in structural systems is not merely a matter of how often failure occurs, but how severely. The index $\betas$ provides engineers with a more sensitive and consequence-aware tool to identify, interpret, and mitigate structural vulnerabilities, especially in contexts where rare events may lead to disproportionate consequences.

\section*{Acknowledgments}
This work was supported by the University of Sharjah, UAE. This support is highly appreciated by
the authors.

\section*{Data Availability Statement}
No external data was used in this research.

\section*{Competing Interests}
The authors declare no competing financial interests or personal relationships that could have appeared to influence the work reported in this paper.

\section*{Author Contributions}
Conceptualization: ML; Methodology: ML; Formal analysis and investigation: ML, SB, RA; Writing - original draft preparation: ML; Writing - review and editing: ML, SB, RA; Programming: ML.

% References
%\bibliographystyle{plain}
\bibliographystyle{unsrt}
\bibliography{references}

\end{document}